\documentclass[letterpaper, 10 pt, conference]{IEEEtran}

\IEEEoverridecommandlockouts 

\usepackage{url}
\usepackage{cite}
\usepackage{hyperref}
\usepackage[utf8]{inputenc}
\usepackage[english]{babel}
\let\labelindent\relax
\usepackage{enumitem}
\usepackage{graphics} 
\usepackage{epsfig} 
\usepackage{mathptmx} 
\usepackage{times} 
\usepackage{amsmath} 
\usepackage{amssymb}  
\usepackage{mathtools}
\usepackage{lipsum}
\usepackage{cancel}
\newtheorem{theorem}{Theorem}
\newtheorem{defn}{Definition}
\newtheorem{obs}{Observation}

\usepackage{algorithm}
\usepackage{algpseudocode}
\usepackage{bbm}

\newenvironment{properties}
 {\enumerate[label=\textbf{P\arabic*.}, ref=P\arabic*]}
 {\endenumerate}
\makeatletter
\newcommand\varitem[1]{\item[\textbf{P\arabic{enumi}\rlap{$#1$}.}]%
  \edef\@currentlabel{A\arabic{enumi}{$#1$}}}
\makeatother

\title{\LARGE \bf
Pricing in Ride-sharing Markets : Effects of network competition and autonomous vehicles
}

\author{Diptangshu Sen$^a$ and Arnob Ghosh$^b$
\thanks{$^{a}$ Diptangshu Sen is a graduate student at Georgia Tech School of Industrial \& Systems Engineering.
        {\tt\small dsen30@gatech.edu}}%
\thanks{$^{b}$ Arnob Ghosh is currently affiliated with the NSF-AI Edge Institute at the Ohio State University.  
        {\tt\small ghosh.244@osu.edu}}%
}

\begin{document}

\maketitle
\thispagestyle{empty}
\pagestyle{empty}

\begin{abstract}
Autonomous vehicles will be an integral part of ride-sharing services in the future. This setting differs from traditional ride-sharing marketplaces because of the absence of the supply side (drivers). However, it has far-reaching consequences because in addition to pricing, players now have to make decisions on how to distribute fleets across network locations and re-balance vehicles in order to serve future demand. In this paper, we explore a duopoly setting in the ride-sharing marketplace where the players have fully autonomous fleets. Each ride-service provider (RSP)'s prices depend on the prices and the supply of the other player. We formulate their decision-making problems using a game-theoretic setup where each player seeks to find the optimal prices and supplies at each node while considering the decisions of the other player. This leads to a scenario where the players' optimization problems are coupled and it is challenging to find the equilibrium. We characterize the types of demand functions (e.g.: linear) for which this game admits an exact potential function and can be solved efficiently. For other types of demand functions, we propose an iterative algorithm to compute the equilibrium. We conclude by providing numerical insights into how different kinds of equilibria would play out in the market when the players are asymmetric. Our numerical evaluations also provide insights into how the regulator needs to consider network effects while deciding regulation in order to avoid unfavorable outcomes. 
\end{abstract}

\section{Introduction}
\subsection{Motivation}
When Uber started operations in San Francisco in 2010, a new era of ride-sharing systems was ushered in. Ride-sharing systems are examples of the `sharing economy' where users log onto a platform (like Uber) to request rides and the platform matches them with potential drivers nearby. Since 2010, the ride-sharing market has taken off rapidly and is expected to be worth $456bn ~\$$ by the end of 2023 (\cite{statista}). The traditional ride-sharing marketplace is two-sided and the ride-service provider (RSP) can control both supply and demand sides by using the price signal. However, competition pushes passenger fares down and driver wages up and reduces the margin for the RSPs.

With significant technological advances over the last few years, autonomous vehicles are the next big step in the ride-sharing marketplace. Lyft has already introduced autonomous fleets in Las Vegas, Miami and more recently in Austin, Texas (\cite{lyft1}, \cite{lyft2}). Uber also has similar plans by the end of 2022 (\cite{uber1}). While such futuristic developments are exciting, they alter the dynamic of the marketplace entirely. The supply side of the marketplace is now non-existent and platforms are required to maintain their own autonomous fleets. Hence, it is profitable for the RSPs to have a such fleet as they do not need to incentivize drivers.  A fixed supply also means that apart from pricing, there are other important considerations to be made, like what should be the optimal fleet size and how to dispatch vehicles optimally. Thus, the decision-making process becomes more convoluted.  In this paper, our goal is to study the following : 
\begin{enumerate}
    \item How can players with autonomous vehicles make decisions about how to price, dispatch, and re-balance effectively? 
    \item How will the above decisions be affected when there is competition in the marketplace? 
    \item What kind of effects would imposing regulations (like parking costs, and congestion taxes) have on the RSPs?
\end{enumerate}
The readers may be wondering why these are questions worth answering. Although the first two questions have been studied extensively for the traditional ride-sharing marketplace, analysis of the case where autonomous vehicles are involved is in its nascent stage. We have already explained how the case with autonomous vehicles is different from the traditional case. Also, there are examples in recent times where replacing human components of systems with autonomous components, can lead to unfavorable consequences (\cite{npr}). So it is imperative that we study these hybrid systems exhaustively before deploying them in the real world. Further, we also need to understand how to regulate those marketplaces to increase consumers' welfare.  Now, we highlight why some of these problems are challenging. \cite{AG2022} has shown that under restrictive assumptions (linearity) on the demand function, there is an easy way to compute the equilibrium of the problem using potential functions. However, the problem becomes challenging when we relax the linearity assumption (Refer Section \ref{sec:equilibrium}). It is also not clear apriori how these hybrid systems will respond to player asymmetries and regulations. We investigate all these aspects of the problem in the paper. 
We answer the questions using a \textit{duopoly setting} where only two platforms are involved in the marketplace. While a duopoly setting might sound restrictive, however we observe that  many major ridesharing markets around the world have evolved into duopolies like Uber-Lyft in the US (\cite{duopoly1}), Uber-Ola in India (\cite{duopoly2}). Further, \cite{AG2022} also considers a duopoly setting. 

\subsection{Contributions}
We summarize the main contributions of our paper below. 
\begin{enumerate}[leftmargin=*]
    \item We consider a generic demand model (linear or non-linear in prices) to capture the dynamics of a ride-sharing marketplace. To the best of our knowledge, the existing literature considers linear demand function which is indeed a restrictive assumption because, in real life, demand functions are estimated using data and are hardly ever linear. 
    \item Our modeling choice makes the problem more challenging because it no longer admits a potential function (which is one of the key solution concepts in network games). So we develop an iterative algorithm to compute the equilibrium which applies to any general demand function. Even though we cannot provide formal convergence guarantees on the algorithm at this time, empirically our algorithm converges to the equilibrium quickly. 
    \item We investigate the properties of the equilibrium in a variety of settings. We show through numerical experiments that when the players in the market are asymmetric and demands are unbalanced (unbalanced demand means that different nodes in the network have significantly different levels of demand), the smaller player may be forced to exit the market meaning that they will not serve certain origin-destination pairs leading to `localized monopolies'.
    
    \item We also show that network effects play significant roles in  selecting regulations. This helps us to generate useful insights on how regulators can design price regulations that can achieve desired outcomes like increasing passenger welfare and decreasing idle vehicles. 
\end{enumerate}
\subsection{Literature Review}
In recent years, the ride-hailing marketplace has been an area of active research. The pricing problem is a consistent theme in the existing literature and two notable papers which explore this are \cite{SB2015} and \cite{KB2019}. \cite{SB2015} uses a queuing theory framework for matching passengers with rides. 
In contrast to studying the temporal variation, \cite{KB2019} considers the effect of spatial variation. There have also been other work which has explored pricing for ride-sharing platforms using tools like reinforcement learning (\cite{VA2021}). However, this line of work is with respect to a single platform (which operates in a traditional two-sided marketplace) and does not consider the competition. Hence, those analyses can not be extended to our setting.

The network effect on a single platform or market maker has been studied \cite{cai2019role,xu2017efficiency,pang2017efficiency,lin2017networked}. In \cite{cai2019role,xu2017efficiency} the market maker or platform procures supply across multiple locations and then transports those supplies to meet the demand. Multiple firms compete at each node as in the Cournot competition model. In \cite{pang2017efficiency,lin2017networked} the platform creates a network by assigning edges between the supply and the demand side where suppliers/firms can only serve the customers connected through edges. However, in our setup, multiple platforms (RSPs) are competing across multiple locations instead of a single platform. Further, the competition model we consider is different from the networked Cournot model considered in the above papers. 


More recently, researchers have considered how to optimally re-balance (sending idle vehicles from one location to serve the demand at other locations) \cite{BV2019}, and selecting prices \cite{CC2020} when the RSP has autonomous vehicles. \cite{CC2021} considers the setup where the RSP selects optimal prices and rebalances the vehicles jointly. However, these papers do not consider the effect of competition of multiple RSPs. There has also been some work that considers competition among  ride-sharing platforms. \cite{NZ2017} explores how competition on the supply side of the marketplace can affect driver wages and passenger welfare. However, as we pointed out earlier, the competition with autonomous vehicles is a different scenario.  

In terms of context, the works which are closest to ours are \cite{AZ2021} and \cite{AG2022}. \cite{AZ2021} looks at price competition in a duopoly setting with two platforms that own fully autonomous fleets. However, it assumes that the platforms have identical operation costs which leads to a symmetric equilibrium (i.e., identical prices for both players for each source-destination pair). However, the firms may differ significantly in terms of fleet sizes and this can lead to asymmetric equilibria which we have considered in our paper, but has not been taken into account in \cite{AZ2021}. This allows us to get insights into  the properties of asymmetric equilibria. As we mentioned earlier, \cite{AG2022} also investigates the competition between two ride-sharing platforms with autonomous fleets.  However, \cite{AG2022} considers a linear model. Compared to both \cite{AZ2021,AG2022} we consider generic demand functions. We also provide an iterative algorithm on how to find equilibrium for these generic demand functions. For a special non-linear demand model, we observe that such an algorithm indeed converges quickly. Further, unlike the above papers, we also investigate the impact of various forms of regulation which guides us on how to attain outcomes that will be beneficial to the passengers.

\section{Modeling}
\subsection{Network}
We consider a simple two-node network which is represented by the complete directed graph $G = (\mathcal{N}, \mathcal{A})$ (Refer Fig. \ref{fig1:graph}). Clearly, $\mathcal{N} = \{1, 2\}$ and $\mathcal{A} = \{e_{11}, e_{12}, e_{21}, e_{22}\}$. We also assume that transit times along all arcs are the same. Note that our analysis and insights go through for larger networks (with $|\mathcal{N}| > 2$) and different arc transit times. The primary reason for considering a small version of the network is to visualize and interpret the network effects on the decisions more meaningfully.  
\begin{figure}
      \centering
      \includegraphics[width=0.45\textwidth]{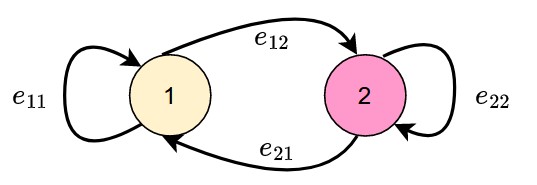}
      \caption{Network Structure}
      \label{fig1:graph}
      \vspace{-0.1in}
\end{figure}

\subsection{Players \& Interactions Model}
There are RSPs $A$ and $B$ who operate in the marketplace (i.e., a duopoly setting). Each RSP has a fixed number of vehicles in its fleet, however, the number may be different across the RSPs (asymmetric players). Vehicles are used in one of three ways: i) they serve passengers, ii) they are routed empty (which we denote as {\em re-balancing} throughout this paper) to places of high expected demand, and iii) they stand idle at one or multiple locations (equivalent to be `parked'). We assume that a vehicle can only serve one passenger in a trip. Each RSP earns revenue by serving demand. They also incur costs for operating the fleet (fuel costs/ `congestion taxes') or for keeping vehicles idle (`parking costs'). Therefore, the RSP has to decide how to set prices and use vehicles judiciously so that they can maximize their profits. Market competition makes the pricing decision more convoluted because players \footnote{We will use the terms ride-service providers, players, and RSPs interchangeably.} now have to take into consideration prices set by the other player. We will capture the effect of one player's prices on the other using the \textit{demand function} which we explain subsequently. The players are assumed to be rational and selfish, so we model their interaction as a \textit{simultaneous}, \textit{non-cooperative game}. Note that our model does not consider any temporal aspects as our focus is on investigating the spatial effect. The characterization of the model for the temporal variation of demand has been left for the future. 

\subsection{Demand Function Modeling}
The demand function outputs the fraction of the market share acquired by an RSP, given its own price and the price of the competitor. Let $f(p_A, p_B)$ represent the demand function for player $A$ which depends on the prices of both players $A$ and $B$, given by $p_A$ and $p_B$ respectively (similarly, player $B$'s demand function would be given by $f(p_B, p_A)$). We assume that all prices are normalized by $P$ and hence are in the interval $[0, 1]$. $P$ can be thought of as the price of an alternative commuting option, so no price in the network should exceed $P$, otherwise, passengers will avail of the outside option. In order to qualify as a candidate demand function, it is necessary that $f(\cdot)$ have the following desirable properties :
\begin{properties}
    \item \label{item:P1} $0 \leq f(p_A, p_B) \leq 1$ : It is not possible to capture a market share that is negative or greater than $1$.
    \item \label{item:P2} $f(p_A, p_B) + f(p_B, p_A) \leq 1$ :  The total market share captured by both players cannot exceed $1$. 
    \item \label{item:P3} $p_A = p_B \implies f(p_A, p_B) = f(p_B, p_A)$ : This means that if both players set equal prices, they should capture identical market shares. It also reflects an inherent assumption in our model that no passenger has specific preferences for a particular ride-service provider. Also, note $f(0, 0) = \frac{1}{2}$. 
    \item \label{item:P4} $f(p, p)$ is non-increasing in $p$ : This captures the price sensitivity of passengers. Even if both players  set the same price $p$, as $p$ increases, fewer and fewer people would be availing of a ride as it exceeds their willingness-to-pay threshold. 
    \item \label{item:P5} $p_A > p_B \implies f(p_A, p_B) \leq f(p_B, p_A)$ : If player $A$ sets a higher price than $B$, then $A$ cannot capture a strictly larger market share than $B$. 
    \item \label{item:P6} $p_A > p_A^{\prime} \implies f(p_A, p_B) \leq f(p_A^{\prime}, p_B)$ : This means that with player $B$'s price fixed, if player $A$ increases its price, then its market share will reduce. If $f(\cdot)$ is differentiable with respect to $p_A$, this condition is equivalent to $\frac{\partial f}{\partial p_A} \leq 0$
    \item \label{item:P7} $p_B < p_B^{\prime} \implies f(p_A, p_B) \geq f(p_A, p_B^{\prime})$ : If the competitor $B$ increases its price, then player $A$'s market share should increase for the same price $p_A$. If $f(\cdot)$ is differentiable with respect to $p_B$, it is equivalent to $\frac{\partial f}{\partial p_B} \geq 0$ 
    \item \label{item:P8} $f(1, p_B) = 0$ : If player $A$ sets $p_A = 1$ (the highest possible price), then its market share goes to \textit{zero}, irrespective of $p_B$. Alternatively, $p_A = 1$ represents a scenario where player $A$ does not compete in the market and $B$ has a monopoly. 
    \item \label{item:P9} $f(0, 1) = 1$ : If player $A$ has a monopoly, then it can capture the whole market by setting $p_A = 0$. This is again intuitive. 
\end{properties}
In recent ride-sharing literature, \cite{AG2022} uses a specific piecewise linear form. It can be easily verified that it satisfies all the aforementioned properties. However, linearity is a very strong assumption because, in reality, demand functions are hardly ever linear. Thus, our goal is to analyze {\em any} demand model which captures all the above properties. As an example of a non-linear model for which we will evaluate our numerical results, we consider the following bi-linear form: 
\begin{equation}\label{eq:demandfunction}
    f(p_A, p_B) = \frac{1}{2}(1-p_A)(1+p_B) \quad \forall ~0 \leq p_A, p_B \leq 1
\end{equation}

\subsection{RSP's Optimization Problem}
Each RSP seeks to maximize its profit. Profit for an RSP is expressed as the difference between total revenues earned and operation costs incurred. Revenues are earned from serving demand in the network. Operation costs are classified into two components :
\begin{enumerate}[leftmargin=*]
    \item \textit{Re-balancing costs} : The RSP has the option to route empty vehicles to other locations where demand is higher. Thus, re-balancing costs may represent fuel costs for the empty cars plying on the network. However, sometimes, re-balancing vehicles can also cause a nuisance by contributing to higher congestion. To prevent such behavior, the central planner may impose penalties on empty routing vehicles. All such penalties are also included in this cost. 
    \item \textit{Parking costs} : If demand is low, the RSP will be forced to keep some of its fleets idle. Sometimes, the RSP might also have the incentive to keep vehicles idling deliberately to create an artificial lack of supply in the market and jack up the price. If the RSP chooses to keep vehicles parked at any node, it has to pay parking costs. For example, during certain periods of the day, the central planner may choose to impose high parking fees at specific locations to specifically avoid idling behavior.  
\end{enumerate}
Each ride-service provider has to make the following decisions : 
\begin{enumerate}[leftmargin=*]
    \item {\em How to choose ride fares $p_{(\cdot)}^{ij}$ for every arc in the network ?} The RSP wants to earn higher revenues, so it may be tempted to set high prices but there is a trade-off because high prices mean that a lower number of passengers would be willing to ride and there is the risk of losing out of market share to the competitor. Pricing is also important for matching supply with demand. Since the supply is limited, pricing too low would mean that demand exceeds supply and many passengers who want to avail of a ride, do not get matched with a vehicle.  
    \item {\em Given a fixed fleet ($m_{\cdot}$) size, how to allocate supplies $m_{(\cdot)}^{i}$ optimally across all the network nodes?} This decision depends on a lot of factors. Intuitively, it is more profitable to allocate larger supplies to nodes with high expected demand. However, it might happen that competition at the other node is low. In case some part of the fleet is not in use and will be kept idle (low demand regime), it may be more favorable for the RSP to place those vehicles at a location that has lower parking costs. 
    \item {\em How to choose re-balancing flows of empty vehicles (given by $r_{(\cdot)}^{ij}$) throughout the network effectively?} Again, it is not apriori straightforward to optimally decide. When rides carry passengers from one location to another, it leads to the accumulation of vehicles at the destination node. If demand along the reverse direction is scarce, the RSP has to re-route some of these empty vehicles to meet demand in other locations because supply is limited. 
\end{enumerate}

\subsubsection{Assumptions}
We make the following assumption :
 We treat vehicles like divisible commodities. This enables us to model them as continuous variables and retain the convexity of the solution space (ideally, they should be modeled as integer variables because the number of passengers or vehicle supply cannot be fractional). It is well-known that mixed-integer models are difficult to solve and solution methods often do not scale well.

\subsubsection{Notation Key}
In this segment, we introduce the notation for our formulation. Unless otherwise specified, this notation applies to the rest of the paper. (Refer Table \ref{tab:notation})
\begin{table}[h!tb]
    \centering
    \begin{tabular}{|c|l|}
    \hline
        \textbf{Notation} & \textbf{Description} \\ \hline
       $m_A, m_B$ & Total fleet sizes of players $A$ and $B$ respectively \\ \hline
       $m_A^{i}, m_B^{i}$ & supply of vehicles at node $i \in \mathcal{N}$\\ \hline
       $p_c^{ij}$ & transit cost per trip along arc $e_{ij} \in \mathcal{A}$ \\ \hline 
       $p_e^{i}$ & parking cost per vehicle at a node $i \in \mathcal{N}$ \\ \hline
       $p_A^{ij}, p_B^{ij}$ & ride fare along arc $e_{ij} \in \mathcal{A}$ \\ \hline
       $x_A^{ij}, x_B^{ij}$ & rate of serving demand along $e_{ij} \in \mathcal{A}$ \\ \hline
       $r_A^{ij}, r_B^{ij}$ & rate of re-balancing along $e_{ij} \in \mathcal{A}$ \\ \hline
       $D^{ij}$ & rate of demand generation along $e_{ij} \in \mathcal{A}$ \\ \hline
    \end{tabular}
    \vspace{0.1in}
    \caption{Notation Summary}
    \label{tab:notation}
    \vspace{-0.4in}
\end{table}


\subsubsection{Formulation}
We now introduce our formulation for player $A$'s optimization problem (given by $\mathcal{F_A}$). Player $B$'s optimization problem will be similar, so we will omit it here. 
\begin{subequations}\label{eq:opt_problem}
    \begin{align}
        \mathcal{F_A} := \max_{p_A^{ij}, r_A^{ij}, m_A^{i}}& ~ \sum_{e_{ij} \in A}(p_A^{ij}-p_c^{ij})x_A^{ij} - \sum_{e_{ij} \in A}p_c^{ij} r_A^{ij} \notag \\
         - &\sum_{i \in N}p_e^{i}\left(m_A^{i} - \sum_{j \in N}(x_A^{ij}+r_A^{ij})\right) \notag \\
        s.t. \quad &x_A^{ij} = D^{ij}\cdot f(p_A^{ij}, p_B^{ij}) \quad \forall~e_{ij} \in \mathcal{A} \label{con:demand} \\
        &\sum_{j \in N}\left(x_A^{ij}+r_A^{ij}\right) \leq m_A^{i} \quad \forall ~i \in \mathcal{N} \label{con:supply}\\
        &\sum_{\substack{j \in N\\j \neq i}}\left(x_A^{ij}+r_A^{ij}\right) = \sum_{\substack{j \in N\\j \neq i}}\left(x_A^{ji}+r_A^{ji}\right) \quad \forall ~i \in \mathcal{N} \label{con:flow}\\
        &\sum_{i \in N} m_A^{i} = m_A \label{con:fleet}\\
        &0 \leq p_A^{ij} \leq 1 \quad \forall ~e_{ij} \in \mathcal{A} \label{con:bound}\\
        &0 \leq r_A^{ij} \quad \forall ~e_{ij} \in \mathcal{A}, \quad 0 \leq m_A^{i} \quad \forall ~i \in \mathcal{N} \label{con:nonneg}
    \end{align}
\end{subequations}
As described earlier, player $A$ seeks to maximize its profit (Revenue - rebalancing costs - parking costs). (\ref{con:demand}) is the demand constraint (note that we can omit this constraint by substituting $x_A^{ij}$'s throughout, however, it makes the problem non-linear in $p_A^{ij}$). (\ref{con:supply}) is a supply constraint at every node because the total outflow rate of vehicles from node $i \in \mathcal{N}$ cannot exceed the supply rate of vehicles at the node. The equality constraint in (\ref{con:flow}) represents a flow balance constraint at each node. This is needed as the total inflow rate must equal to the total outflow rate from any given node. Observe that we deliberately omit flows along the self-loops (edges $e_{ii}$) because it cancels out on both sides of the equation. The next constraint (\ref{con:fleet}) implies that the total supply across the network must add up to the fleet size of the RSP. Finally, we have the bounds on the prices (\ref{con:bound}) and the non-negativity constraints (\ref{con:nonneg}). Since prices are normalized, they cannot exceed $1$. 
\subsubsection{Essential Insights}
Even before we proceed to finding the equilibrium of the game, several interesting observations can be made about player $A$'s optimization problem $\mathcal{F_A}$. 
\begin{obs}\label{obs:unique}
When player $B$'s prices are known and we consider the form of the demand function in Eq.\ref{eq:demandfunction}, player $A$'s optimization problem admits a unique solution (and vice-versa).
\end{obs}
\begin{proof}
This is easy to verify. When we use the demand function form in Eq. \ref{eq:demandfunction}, the objective function is concave in $p_A^{ij}$ (quadratic in $p_A^{ij}$ with negative leading co-efficient), $r_A^{ij}$ and $m_A^{i}$. Also, the feasible set is convex and compact. Maximizing a concave function over a convex compact set always leads us to a unique solution. 
\end{proof}
Note that Observation \ref{obs:unique} will be used later to ensure the uniqueness of the equilibrium in Section \ref{sec:alg} for any given price vector $p_B$ and the demand function form in Eq. \ref{eq:demandfunction}.

Note that if $f(\cdot)$ is non-linear in $p_A$, the constraint in (\ref{con:demand}) makes the problem non-convex for a given price in $p_B$, hence, we can not guarantee the above observation. However, we can still use any non-linear solver to find $p_A$ for a given price of the other player $B$.

The next observation tells us that even if the lower bound of the price is $0$, under a competitive setup, no player would choose a price smaller than $1/2$.
\begin{obs}   \label{obs:pricebound} 
All prices on the network will be $\geq \frac{1}{2}$.
\end{obs}
\begin{proof}
Observe that $p_A^{ij} \geq p_c^{ij}$ always. This is because setting a price lower than $p_c^{ij}$ would lead to negative revenues from serving demand on $e_{ij}$ and is not favorable to the RSP. $p_A^{ij} = 1$ is equivalent to player $A$ not competing on $e_{ij}$ at all, so this case is not very interesting. (Also, when $p_A^{ij} = 1$, the observation is trivially correct.) So, we now look at cases when $p_c^{ij} \leq p_A^{ij} < 1$. Using the KKT-conditions on Problem \ref{eq:opt_problem}, we can show that :
    \begin{equation*}
        p_A^{ij} = \frac{1}{2}(1+Q^{ij})
    \end{equation*}
where $Q^{ij}$ is the Lagrange multiplier associated with the non-negativity constraint $r_A^{ij} \geq 0$. Since $Q^{ij} \geq 0$, this implies $p_A^{ij} \geq \frac{1}{2}$. Equality holds when $Q^{ij} = 0$, that is whenever $r_A^{ij} > 0$ (using complementary slackness). For a detailed analysis of the KKT-conditions, please refer to Appendix section of the paper.
\end{proof}

\textit{Remarks : } Although our main focus in this paper is to study competition in the duopoly setting, note that our model can also be used to simulate monopoly scenarios in the ride-sharing marketplace. This can be done by simply setting the price vector of one of the players equal to $\mathbbm{1}$. By property \ref{item:P8} of the demand function, it ensures that the market share of this player goes to \textit{zero}, irrespective of the price point of the other player. This is very convenient because it allows us to compare the properties of the monopoly market with the competitive market.

\section{Computing Equilibria}\label{sec:equilibrium}
From player $A$'s optimization problem in Eq. \ref{eq:opt_problem}, it is clear that $A$ cannot solve for the optimal decisions without considering the prices set by player $B$. Similarly, player $B$'s decisions are also dependent on prices set by $A$. Thus, their optimization problems are coupled. In this setting, we define an \textit{equilibrium} as the combination of decisions $\left(\{p_A^{ij}\}, \{r_A^{ij}\}, \{m_A^{i}\}, \{p_B^{ij}\}, \{r_B^{ij}\}, \{m_B^{i}\}\right)$ from where neither player $A$ nor player $B$ has the incentive to deviate unilaterally. We assume a complete information setup where both players have knowledge of the other player's fleet sizes. The characterization of equilibria for the incomplete information setup is an interesting future direction and is out of the scope of this paper. 

\subsection{Potential functions}
Nash equilibrium in general is difficult to obtain computationally even if a game admits equilibrium. Potential functions are extremely useful tools to compute the equilibria of multi-player games. They involve finding a global potential function $\Phi(\cdot)$ which tracks the change in payoffs whenever any player unilaterally changes their strategy. It has been established that the local optima of the potential function correspond to the pure Nash equilibria of the underlying game. \cite{AG2022} has already shown that there exists an exact potential game corresponding to a 2-player duopoly setting in the ride-sharing marketplace with fully autonomous vehicles, for a specific choice of the demand function. Our aim is to investigate if and when the potential game approach can be applied more generally.
\begin{obs}
For the bi-linear demand function defined in Eq. \ref{eq:demandfunction}, this game does not have an exact potential function.
\end{obs}
\begin{proof}
The proof is by contradiction. Details can be found in the Appendix section of the paper.  
\end{proof}
Now, we have already seen an example where if the structure considered in \cite{AG2022} is not satisfied, we may not have a potential function. Naturally, this leads to the question: Under what conditions on the demand function $f(\cdot)$ can the game admit a potential function?
We take a small step towards answering this question which leads to the following result. 
\begin{theorem}\label{thm1}
When the demand function $f(p_A, p_B)$ is of the form $f(p_A, p_B) = g(p_A) + h(p_B)$, the game does have an exact potential function if and only if $h(p_B)$ is a linear function in $p_B$ (that is, $h(p_B) = Cp_B$ for some $C > 0, C \in \mathbb{R}$). 
\end{theorem}
\begin{proof}
Detailed proof can be found in the Appendix. 
\end{proof}
Note that when the game admits a potential function, we can solve the potential function to obtain the equilibrium as has been done in \cite{AG2022}. 

\textit{Remark :} It is easy to see that the demand function used in \cite{AG2022} is a special case of the functional form in Theorem \ref{thm1} where $g(p_A)$ is linear in $p_A$. Hence, it must admit an exact potential function. 

\subsection{Algorithm for Finding Equilibrium}\label{sec:alg}
In this segment, we propose a heuristic algorithm (Refer Algorithm \ref{alg}) to compute the equilibrium of the game for general demand functions which do not admit potential functions.  In particular, we consider that the demand function is of the non-linear form in Eq. \ref{eq:demandfunction}. Unfortunately, we are unable to provide formal convergence guarantees on the algorithm at this time. However, empirically, we find that the algorithm converges quickly to the equilibrium. 

We introduce some more notation here. Let $\mathcal{F_A}$ and $\mathcal{F_B}$ be the optimization problems of players $A$ and $B$ respectively. Now,  
\begin{defn}
$\mathcal{BR}(\mathcal{F_A} ~|~p_B)$ represents the best response price of player $A$ to the price $p_B$ set by player $B$. Similarly, for player $B$, we denote $\mathcal{BR}(\mathcal{F_B} ~|~p_A)$ as the best response price of player B when player A sets its price $p_A$.
\end{defn}
 Note that finding the best response price given the other player's price just involves solving a convex optimization problem and hence, can be done very efficiently using any standard convex-opt solver. 
\begin{algorithm}\label{alg}
\caption{Iteratively Computing the Equilibrium}\label{alg}
\begin{algorithmic}
\Require initializing price vector $p_B^{(0)}$
\Require tolerance level $\epsilon > 0$
\State $p_A^{(0)} \gets \mathcal{BR}(\mathcal{F_A} ~|~ p_B^{(0)})$
\State $p_B^{(1)} \gets \mathcal{BR}(\mathcal{F_B} ~|~ p_A^{(0)})$
\State $p_A^{(1)} \gets \mathcal{BR}(\mathcal{F_A} ~|~ p_B^{(1)})$
\While{$||p_A^{(k)} - p_A^{(k-1)}||_{\infty} > \epsilon$ or $||p_B^{(k)} - p_B^{(k-1)}||_{\infty} > \epsilon$}
\State $p_B^{(k+1)} \gets \mathcal{BR}(\mathcal{F_B} ~|~ p_A^{(k)})$
\State $p_A^{(k+1)} \gets \mathcal{BR}(\mathcal{F_A} ~|~ p_B^{(k+1)})$
\State $k \gets k+1$
\EndWhile
\end{algorithmic}
\end{algorithm}

\textbf{The high-level idea}: 
We can think of $\mathcal{BR}(OPT_{A} ~|~p_B)$ as a function $\Theta(p_B)$ and similarly $\mathcal{BR}(OPT_{B} ~|~p_A)$ as a function $\zeta(p_A)$. Then from the iteration algorithm, we have $p_A^{(k)} = \Theta(p_B^{(k)})$ and $p_B^{(k)} = \zeta(p_A^{(k-1)})$ which implies $p_A^{(k)} = \Theta \circ \zeta (p_A^{(k-1)})$. 

\textbf{Why is it difficult to show convergence guarantee?}
Proving convergence of the algorithm is equivalent to showing that the function composition $\Theta \circ \zeta$ has a fixed point and by symmetry, $\zeta \circ \Theta$ also has a fixed point. However, that requires us to show that $\zeta(\cdot)$ and $\theta(\cdot)$ are continuous mappings which is difficult in this setting. 

\textbf{Comment on convergence from empirical evaluations}

\begin{enumerate}[leftmargin=*]
\item Even though we can not show theoretical convergence, the algorithm was empirically found to converge quickly to the equilibrium, usually within 5 iterations for $\epsilon = 0.01$. 
\item For a given price vector of the other player (say $B$), the algorithm  will always produce a \textit{unique solution} to the player $A$. This follows directly from Observation \ref{obs:unique}.
    \item The algorithm was empirically found to converge to the same equilibrium, irrespective of the initial points. 
    
\end{enumerate} 

\textit{Remark}: For non-linear functions other than Eq.~\ref{eq:demandfunction}, we can still apply this heuristic algorithm. However, since the optimization problem for a player may not be convex (given the price of the other player), we need to rely on a non-linear optimization problem solver and we can not guarantee uniqueness of the equilibrium.


\section{Numerical Experiments}
In this section, we will explore the properties of the equilibrium in a variety of settings using numerical experiments. We will start by describing the simulation setup for the experiments. 
\subsection{Simulation Setup}
We fix the total supply on the network to be $\mathcal{S} = 1000$ for all our experiments. We use a parameter $m$ to determine the total demand $\mathcal{D}$ across the network according to the relation $\mathcal{D} = m \mathcal{S}$. For example, $m = 0.5$ indicates that $\mathcal{D} = 500$ and it is a low-demand regime. $m = 1$ indicates that total demand and supply are exactly matched while $m > 1$ indicates high demand regimes. The distribution of demand across the different arcs is controlled by another parameter $\alpha$ (explained in detail in Section \ref{demand_patterns}). We also have a parameter $\beta$ which determines what fraction of total supply $\mathcal{S}$ is owned by which player. Player $A$ has a fleet of size $\beta \mathcal{S}$ while player $B$ has a fleet of size $(1-\beta)\mathcal{S}$. $\beta = 0.5$ indicates that both players are symmetric while $\beta < 0.5$ indicates that $B$ is the larger player in the market (vice-versa for $\beta > 0.5$). 
\begin{table}[h!tb]
    \centering
    \begin{tabular}{|c|c|c|}
    \hline 
    Parameter & Range & Description \\ \hline
    $S$ & $1000$ & Total Supply on network \\ \hline
    $\mathcal{D}$ & $Depends~on~m$ & Total Demand on network \\ \hline
    $m$ & $0.5 - 2.0$ & Demand multiplier wrt $\mathcal{S}$ \\ \hline 
    $\beta$ & $0.2 - 0.5$ & Fraction of $\mathcal{S}$ owned by $A$ \\ \hline 
    $\alpha$ & $0.0 - 1.0$ & Demand distribution parameter \\ \hline
    \end{tabular}
    \vspace{0.1in}
    \caption{Simulation Parameters}
    \label{tab:param}
    \vspace{-0.4in}
\end{table}

\subsection{Demand Patterns}\label{demand_patterns}
We consider the following demand patterns in our experiments: 
\begin{enumerate}[leftmargin=*]
    \item \textit{Pattern 1} : Here, the demand is equally split between the two nodes. We use control parameter $\alpha$ to choose how much of the demand at each node flows towards node $1$. $\alpha$ is varied in the range $[0.5, 1]$. $\alpha = 0.5$ represents a perfectly balanced network while $\alpha = 1$ represents the scenario where all demand is concentrated on arcs which end in node $1$. Let $\mathcal{D}$ represent the total demand across the network. Then the demand distribution can be represented by matrix $D$ where $D_{ij}$ is the demand on arc $e_{ij}$ :
    \[
    D = \begin{bmatrix}0.5\alpha \mathcal{D} & 0.5(1-\alpha) \mathcal{D}\\0.5\alpha \mathcal{D} & 0.5(1-\alpha) \mathcal{D} \end{bmatrix}
    \]
    The case with $\alpha$ close to $1$ may be similar to the busy downtown area ($node~1$) of a city which attracts most of the traffic in the network. 
    \item \textit{Pattern 2} : This demand pattern is used to model scenarios where demand originates out of a single node increasingly with increase in $\alpha$. At $\alpha = 0.5$, the demand is perfectly balanced, but at $\alpha = 1$, all demand originates out of node $1$ and is equally split between arcs $e_{11}$ and $e_{12}$. The demand matrix $D$ is as follows : 
    \[
    D = \begin{bmatrix}0.5\alpha \mathcal{D} & 0.5\alpha \mathcal{D}\\0.5(1-\alpha) \mathcal{D} & 0.5(1-\alpha) \mathcal{D} \end{bmatrix}
    \]
    One possible example of this type of demand pattern at $\alpha$ close to $1$ could be evening traffic from a busy office area. 
    \item \textit{Pattern 3} : This demand pattern captures two extreme scenarios. When $\alpha = 0$, demand is localized only along the cross-arcs $e_{12}$ and $e_{21}$. When $\alpha = 1$, all demand is split equally between the two self-looping arcs and represents a setting where there is no network effect at all.  The demand matrix $D$ is given by : 
    \[
    D = \begin{bmatrix}0.5\alpha \mathcal{D} & 0.5(1-\alpha) \mathcal{D}\\0.5(1-\alpha) \mathcal{D} & 0.5\alpha \mathcal{D} \end{bmatrix}
    \]
    This could be similar to a scenario where each node represents a professional hub. During the daytime, people travel mostly between the two places for work ($\alpha = 0$), but at night, traffic is localized at the individual nodes ($\alpha = 1$).
    \begin{figure}
    \begin{minipage}{0.23\textwidth}
    \includegraphics[width=0.9\textwidth]{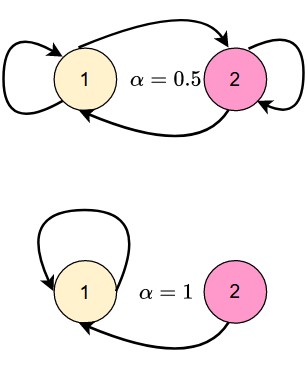}
    \vspace{-0.1in}
    \end{minipage}
    \begin{minipage}{0.23\textwidth}
    \includegraphics[width=0.9\textwidth]{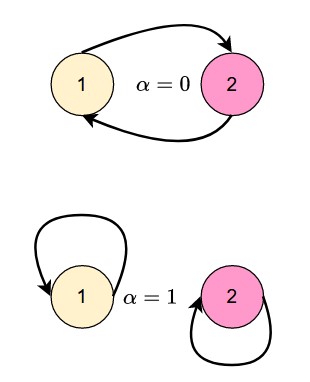}
    \end{minipage}
    \caption{Demand Patterns 1 (left) and 3 (right)}
    \vspace{-0.1in}
    \end{figure}
\end{enumerate}

\subsection{Duopoly Setting}

\subsubsection{Symmetric Players}
For this setting, we have $\beta = 0.5$, so each player has a fleet size of $500$ autonomous vehicles. We study the effects of the demand multiplier $m$ and the demand distribution parameter $\alpha$ on player profits. 
\begin{itemize}[leftmargin=*]
    \item \textit{Nature of the equilibrium : } When the players are symmetric, we end up with a \textit{symmetric equilibrium} where both players choose identical prices, supply patterns and rebalancing flows. This observation is intuitive and in line with the findings of \cite{AZ2021} and \cite{AG2022}.   
    \item \textit{Effect of $m$ :} As the multiplier $m$ increases, player profits are found to increase across all demand patterns. This can be attributed to higher price points and higher utilization of vehicles as $m$ increases (Refer Fig. \ref{fig:variations}). Intuitively, for small values of $m$ (like $m = 0.5$), the demand is much smaller compared to supply, so the players cannot operate at full capacity which leads to costs either in terms of re-balancing or parking costs. Additionally, the smaller demand forces the players to lower their prices because higher prices will alienate most of the passengers. However, in the high demand regime ($m = 2$), players have the flexibility to set high prices and still capture a sizeable portion of the market. Also, their fleets are operating close to full capacity which leads to higher revenues and lower incurred costs. 
    \begin{figure}[!ht]
      \centering
      \includegraphics[width=0.45\textwidth]{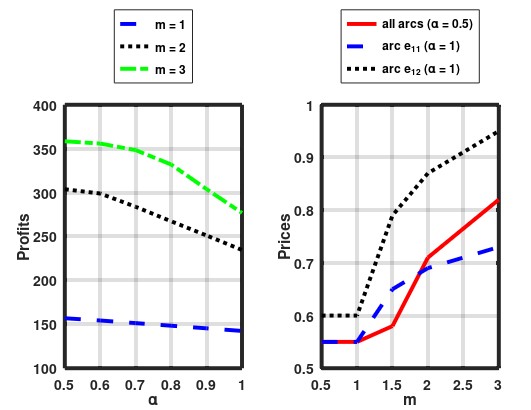}
      \caption{Variations with $m$ and $\alpha$ for demand pattern 2. For the plot on the right, when $\alpha = 0.5$, we report only one price because all edges have same price. When $\alpha = 1$, we report $p_{11}$ and $p_{12}$ because other arcs have zero demand.}
      \label{fig:variations}
      \end{figure}
    \item \textit{Effect of $\alpha$ :} For demand patterns $1$ and $2$, as we increase $\alpha$ from $\frac{1}{2}$ gradually to $1$, the demand across the network becomes \textbf{unbalanced}. This leads to lowering of player profits. This observation is aligned with the findings of \cite{KB2019}. One possible reason for this finding is that as demand gets more unbalanced, the operating inefficiencies for the players increase. For example, let us consider the case with $\alpha = 1$ for pattern $2$. All the demand is restricted to arcs $e_{11}$ and $e_{12}$. So, many vehicles which serve customers along $e_{12}$ are forced to re-route empty back to node $1$ (Refer Fig. \ref{fig3:rebalancing}). Re-balancing incurs costs and reduces the profits. However, in demand pattern $3$, player profits remain unchanged with variation in $\alpha$. This is because there is no need for re-balancing (the demand in the cross-arcs match exactly). So, the player just has to choose the steady state supply at each node optimally. 
    \begin{figure}[!ht]
      \centering
      \includegraphics[width=0.45\textwidth]{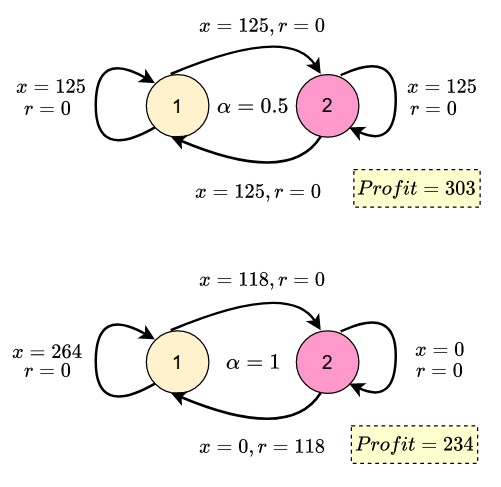}
      \caption{Demand Pattern $1$ under $m = 2$. As $\alpha$ increases from $0.5$ to $1$, demand gets restricted to $e_{11}$ and $e_{12}$. $r_{21}$ increases rapidly leading to decline in profits.}
      \label{fig3:rebalancing}
      \vspace{-0.1in}
\end{figure}
\end{itemize}

\subsubsection{Asymmetric Players}
For this setting, we will look at values of $\beta < 0.5$, so player $A$ has a smaller fleet than player $B$. Note that we do not consider $\beta > 0.5$ because it is identical to the $(1-\beta)$ scenario (where the identities of players are reversed).  
\begin{itemize}[leftmargin=*]
    \item \textit{Big player, Big profits} : We find that as we decrease $\beta$ from $0.5$ (asymmetry increases), player $B$'s profits increase. The increase in profit is much more significant in the higher demand regime ($m = 2$) compared to the low demand regime ($m = 0.5$). This outcome is expected because at high demand, owning a larger fleet provides player $B$ with a large competitive advantage over player $A$. With increase in asymmetry, the market approaches a monopoly market for player $B$. It is well-known that monopoly markets are inefficient, so intuitively it appears that the total size of the market served would decrease with decrease in $\beta$. However, {\em surprisingly}, that does not seem to be the case. Player $B$'s market share definitely increases when it has a larger fleet at its disposal, but the total market served remains unchanged with changes in $\beta$.
    \item \textit{Forced market exits : } Asymmetry creates other adverse effects. In the high demand regime, the player with the smaller fleet might be forced to exit certain markets because serving out of multiple locations may not be profitable with a small fleet. This gives rise to `localized monopolies' where the larger player has unilateral control. This leads to high prices and is not favorable for the passengers. Refer to Fig. \ref{fig4:market_exit} where we identify a scenario where this phenomenon is observed. 
    \begin{figure}[!ht]
      \centering
      \includegraphics[width=0.5\textwidth]{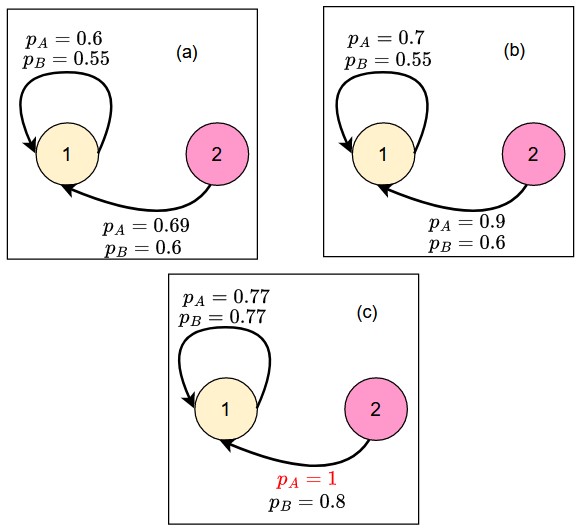}
      \caption{Setting with $\alpha = 1$ and $\beta = 0.2$ for demand pattern 1. (a) $m = 0.5$, (b) $m = 1.0$, (c) $m = 2.0$. Player $A$ forced out on $e_{21}$ for $m = 2.0$. }
      \label{fig4:market_exit}
      \vspace{-0.1in}
    \end{figure}
\end{itemize}

\subsection{Regulations}
In this segment, we will explore how different forms of regulation affect the dynamics of the marketplace. We will specifically look at \textit{two} types of price regulation : 
\begin{enumerate}
    \item \textit{Parking costs at nodes}
    \item \textit{Penalties on re-balancing vehicles}
\end{enumerate}
Also, note that we will study regulations only in the high demand regime ($m = 2$). The purpose of regulations is primarily to prevent strategic behavior from players (like unnecessary re-balancing of empty vehicles) that leads to unfavorable outcomes for passengers or the network as a whole (lesser passengers served, more congestion on the network). That is why regulations generally would not be imposed in the low demand regime (because demand $<$ supply and fleets are already under-utilized, incentive to engage in strategic behavior is small) and there is no motivation to study those scenarios.  
\begin{itemize}[leftmargin=*]
    \item \textit{Regulations affect players' profits disparately : } We find that when players are highly asymmetric (small $\beta$) and demand is unbalanced, regulations affect the larger player's profits much more than the smaller player. Since the regulations are in the form of high parking costs or penalties for re-balancing vehicles, they have negligible effect on the smaller player because it has no re-balancing or parked vehicles (its entire fleet is serving demand). However, regulations do not necessarily lead to passenger welfare because prices increase and number of rides completed, decreases. Refer Tables \ref{tab:noreg} and \ref{tab:reg}. \\
    
    Consider the scenario with demand pattern $2$ and $\alpha = 1$. This means that there is high demand originating in node $1$. The intuition here is that players would try to capture the high demand at node $1$ by re-routing empty vehicles from node $2$ to node $1$ or keep a large supply at node $1$. So, one possible form of regulation could be to increase parking costs at node $1$ and penalize empty vehicles on $e_{21}$. We simulate the scenario with and without regulation. We set $\beta = 0.2$.
    \begin{table}[h!tb]
    \centering
    \begin{tabular}{|c|c|c|c|c|}
    \hline 
    Metrics & $e_{11}$ & $e_{12}$ & $e_{21}$ & $e_{22}$ \\ \hline
    \#Rides completed (A) & 200 & 0 & - & -\\ \hline 
    \#Rides completed (B) & 200 & 200 & - & -\\ \hline 
    Prices (A) & 0.77 & 1.00 & - & - \\ \hline
    Prices (B) & 0.77 & 0.8 & - & - \\ \hline
    \end{tabular}
    \vspace{0.1in}
    \caption{Results without regulations in place, player $A$'s profit = $134$ and player $B$'s profit = $254$. }
    \label{tab:noreg}
    \vspace{-0.2in}
    \end{table}
    
    \begin{table}[h!tb]
    \centering
    \begin{tabular}{|c|c|c|c|c|}
    \hline 
    Metrics & $e_{11}$ & $e_{12}$ & $e_{21}$ & $e_{22}$ \\ \hline
    \#Rides completed (A) & 200 & 0 & - & -\\ \hline 
    \#Rides completed (B) & 200 & 150 & - & -\\ \hline 
    Prices (A) & 0.77 & 1.00 & - & - \\ \hline
    Prices (B) & 0.77 & 0.85 & - & - \\ \hline
    \end{tabular}
    \vspace{0.1in}
    \caption{Results with regulations in place, player $A$'s profit = $134$ and player $B$'s profit = $157$. Prices increase and \# rides completed decreases.}
    \label{tab:reg}
    \vspace{-0.2in}
    \end{table}
    \item \textit{Network effects of regulation :} One consistent pattern that we observe in our numerical experiments is that the larger player does not use all its fleet to serve customers even when the total demand across the network is high. This means that a lot of passengers go unserviced which is clearly not a favorable outcome for social welfare. In this segment, we investigate if regulations can be used appropriately to compel the larger player to serve more demand instead of parking/re-balancing vehicles deliberately. Since we operate in the high demand regime ($m = 2$), price regulations have negligible effect on the smaller player. \\
    
    Consider the following setup of demand pattern $1$ with $\alpha = 0.75$. Set $\beta = 0.2$ which indicates that the players are highly asymmetric and player $B$'s fleet is $4$ times the size of player $A$'s fleet. We re-define $p_c^{ij}$ slightly to include a regulatory component: 
    \begin{equation*}
        p_c^{ij} = p_b^{ij} + v^{ij}
    \end{equation*}
    where $p_b^{ij}$ represents the transit cost on the edge (that applies to all vehicles) while $v^{ij}$ is a penalty that applies only to re-balancing vehicles on $e_{ij}$. For now, we set $p_b^{ij} = 0.1$ for all $e_{ij} \in \mathcal{A}$. When there is no regulation, $v^{ij} = 0 ~\forall~e_{ij}$. We report our solutions in the form of a table. Variables/parameters which are arc-dependent are reported as $2 \times 2$ matrices while node-dependent quantities are reported as vectors of size $2 \times 1$.   \\ 
    \textbf{Scenario 1 (No regulation) :}
    \begin{table}[h!tb]
    \centering
    \begin{tabular}{|c|c|}
    \hline 
    Variables/Parameters & Values \\ \hline
    Parking cost ($p_e$) & $[0; 0]$ \\ \hline
    Penalty for re-balancing ($v$) & $[[0, 0]; [0, 0]]$ \\ \hline
    Supply ($m_B$) & $[495; 305]$ \\ \hline 
    Prices ($p_B$) & $[[0.71, 0.5]; [0.73, 0.55]]$\\ \hline 
    Rides ($x_B$) &  $[[200, 111]; [200, 105]]$\\ \hline
    Rebalancing ($r_B$) & $[[0, 89]; [0, 0]]$\\ \hline
    Idling Vehicles & $[95; 0]$ \\ \hline
    \end{tabular}
    \vspace{0.1in}
    \caption{Scenario 1}
    \label{tab:networkeff1}
    \vspace{-0.2in}
    \end{table}
    Observe that player $B$ has a large number of vehicles ($95$) idling at node $1$ (Table \ref{tab:networkeff1}). Total supply at node $1$ is $495$, out of which $200 + 111 = 311$ serve passengers while $89$ are routed to node $2$ for re-balancing. All the supply at node $2$ is used to serve demand. \textit{The most intuitive regulation here is to impose parking costs at node $1$}. \\ \\
    \textbf{Scenario 2 :}
    \begin{table}[h!tb]
    \centering
    \begin{tabular}{|c|c|}
    \hline 
    Variables/Parameters & Values \\ \hline
    Parking cost ($p_e$) & $[0.5; 0]$ \\ \hline
    Penalty for re-balancing ($v$) & $[[0, 0]; [0, 0]]$ \\ \hline
    Supply ($m_B$) & $[400; 400]$ \\ \hline 
    Prices ($p_B$) & $[[0.71, 0.5]; [0.73, 0.55]]$\\ \hline 
    Rides ($x_B$) &  $[[200, 111]; [200, 105]]$\\ \hline
    Rebalancing ($r_B$) & $[[0, 89]; [0, 0]]$\\ \hline
    Idling Vehicles & $[0; 95]$ \\ \hline
    \end{tabular}
    \vspace{0.1in}
    \caption{Scenario 2}
    \label{tab:networkeff2}
    \vspace{-0.2in}
    \end{table}
    We have now increased the parking cost at node $1$ to $0.5$ (Table \ref{tab:networkeff2}). As we can see, it has no impact on player $B$ whatsoever in terms of prices or passengers served. Player $B$ only adjusts the supply in such a way that the idling vehicles from earlier are now parked at node $2$. This also highlights that \textit{localized regulations may not always produce the desired effect.} \\ \\
    \textbf{Scenario 3 : } We now impose an additional regulation in terms of parking costs at node $2$. This does force player $B$ to use some of the extra supply for serving demand along $e_{22}$ (rides served increases from $105$ to $116$). However, this is not a good outcome because it has created a scenario where the number of vehicles moving around without any passengers has increased significantly. The intuition here is that the RSP finds it more favorable to route the extra supply on the network (because there are no penalties currently) than actually serve demand.  \\
    \begin{table}[h!tb]
    \centering
    \begin{tabular}{|c|c|}
    \hline 
    Variables/Parameters & Values \\ \hline
    Parking cost ($p_e$) & $[0.5; 0.5]$ \\ \hline
    Penalty for re-balancing ($v$) & $[[0, 0]; [0, 0]]$ \\ \hline
    Supply ($m_B$) & $[463; 337]$ \\ \hline 
    Prices ($p_B$) & $[[0.71, 0.5]; [0.73, 0.5]]$\\ \hline 
    Rides ($x_B$) &  $[[200, 111]; [200, 116]]$\\ \hline
    Rebalancing ($r_B$) & $[[43, 109]; [21, 0]]$\\ \hline
    Idling Vehicles & $[0; 0]$ \\ \hline
    \end{tabular}
    \vspace{0.1in}
    \caption{Scenario 3}
    \label{tab:networkeff3}
    \vspace{-0.3in}
    \end{table}
\end{itemize}
The key takeaways from this discussion are as follows : 
\begin{enumerate}[leftmargin=*]
    \item Regulations do not affect all players in the same way. They affect the larger player in the market significantly more than the smaller player. 
    \item Localized regulations often may not have the desired outcome. This is primarily due to network effects. Since the RSP has control over its whole supply, it can circumvent local price regulations by diverting supplies in the most cost-effective way possible. 
    \item This also gives us some insights into how we can design effective regulations for ride-sharing marketplaces with autonomous vehicles. The regulations need to be coordinated over the entire network to achieve the best possible outcome.  
\end{enumerate}

\section{Conclusion \& Future Scope}
In this paper, we use game theory to study the networked competition of two players in the ride-sharing marketplace where the players have fully autonomous fleets. We propose a non-linear demand function that captures the effect of one player's price on the other. Unlike linear demand functions in literature, we show that the non-linear form does not admit a potential function, so it is technically challenging to solve. We propose an iterative algorithm to compute the equilibrium of the game for any general form of the demand function. We use our model to empirically study properties of the equilibrium under a variety of settings like asymmetric competition and price regulations and develop insights that can help regulators design informed policies/regulations for these markets that can achieve desired outcomes like increased passenger welfare. \\

There are several interesting avenues of future work. One immediate extension is to investigate the equilibrium of this game under incomplete information settings. It may also be worthwhile to explore if it is possible to provide formal convergence guarantees or find necessary conditions for the convergence of our iteration algorithm. In our work, we highlight why it is important to coordinate regulations across the network. Designing effective price regulations that take network effects into account, could be another interesting direction of work.

\section{Acknowledgement}
We thank Dr. Parinaz Naghizadeh at the Ohio State University for her insightful inputs at different stages of the work. 

\bibliographystyle{IEEEtran}
\bibliography{mybib.bib}

\begin{thebibliography}{10}
\providecommand{\url}[1]{#1}
\csname url@samestyle\endcsname
\providecommand{\newblock}{\relax}
\providecommand{\bibinfo}[2]{#2}
\providecommand{\BIBentrySTDinterwordspacing}{\spaceskip=0pt\relax}
\providecommand{\BIBentryALTinterwordstretchfactor}{4}
\providecommand{\BIBentryALTinterwordspacing}{\spaceskip=\fontdimen2\font plus
\BIBentryALTinterwordstretchfactor\fontdimen3\font minus
  \fontdimen4\font\relax}
\providecommand{\BIBforeignlanguage}[2]{{%
\expandafter\ifx\csname l@#1\endcsname\relax
\typeout{** WARNING: IEEEtran.bst: No hyphenation pattern has been}%
\typeout{** loaded for the language `#1'. Using the pattern for}%
\typeout{** the default language instead.}%
\else
\language=\csname l@#1\endcsname
\fi
#2}}
\providecommand{\BIBdecl}{\relax}
\BIBdecl

\bibitem{statista}
Statista. (2023, Feb.)

\bibitem{lyft1}
NYT. (2022, Aug.) Lyft unveils self-driving car service in las vegas (with
  caveats).

\bibitem{lyft2}
Statesman. (2022, Sep.) You can now take a driverless lyft in austin. here's
  what you need to know.

\bibitem{uber1}
Bloomberg. (2022, Oct.) Uber revives self-driving taxi dreams, plans to start
  this year.

\bibitem{npr}
A.~Shahani. (2015, Aug.) Tesla model s can be hacked, and fixed (which is the
  real news).

\bibitem{AG2022}
A.~Ghosh and R.~Berry, ``Competition among ride service providers with
  autonomous vehicles,'' in \emph{2022 20th International Symposium on Modeling
  and Optimization in Mobile, Ad hoc, and Wireless Networks (WiOpt)}.\hskip 1em
  plus 0.5em minus 0.4em\relax IEEE, 2022, pp. 209--216.

\bibitem{duopoly1}
Global-News. (2019, May) Uber, lyft's huge capital have created duopoly in
  rideshare market.

\bibitem{duopoly2}
Business-Wire. (2020, May) Online taxi services market in india.

\bibitem{SB2015}
S.~Banerjee, C.~Riquelme, and R.~Johari, ``Pricing in ride-share platforms: A
  queueing-theoretic approach,'' \emph{Available at SSRN 2568258}, 2015.

\bibitem{KB2019}
K.~Bimpikis, O.~Candogan, and D.~Saban, ``Spatial pricing in ride-sharing
  networks,'' \emph{Operations Research}, vol.~67, no.~3, pp. 744--769, 2019.

\bibitem{VA2021}
M.~Haliem, G.~Mani, V.~Aggarwal, and B.~Bhargava, ``A distributed model-free
  ride-sharing approach for joint matching, pricing, and dispatching using deep
  reinforcement learning,'' \emph{IEEE Transactions on Intelligent
  Transportation Systems}, vol.~22, no.~12, pp. 7931--7942, 2021.

\bibitem{cai2019role}
D.~Cai, S.~Bose, and A.~Wierman, ``On the role of a market maker in networked
  cournot competition,'' \emph{Mathematics of Operations Research}, vol.~44,
  no.~3, pp. 1122--1144, 2019.

\bibitem{xu2017efficiency}
Y.~Xu, D.~Cai, S.~Bose, and A.~Wierman, ``On the efficiency of networked
  stackelberg competition,'' in \emph{2017 51st Annual Conference on
  Information Sciences and Systems (CISS)}.\hskip 1em plus 0.5em minus
  0.4em\relax IEEE, 2017, pp. 1--6.

\bibitem{pang2017efficiency}
J.~Z. Pang, H.~Fu, W.~I. Lee, and A.~Wierman, ``The efficiency of open access
  in platforms for networked cournot markets,'' in \emph{IEEE INFOCOM 2017-IEEE
  Conference on Computer Communications}.\hskip 1em plus 0.5em minus
  0.4em\relax IEEE, 2017, pp. 1--9.

\bibitem{lin2017networked}
W.~Lin, J.~Z. Pang, E.~Bitar, and A.~Wierman, ``Networked cournot competition
  in platform markets: Access control and efficiency loss,'' \emph{ACM
  SIGMETRICS Performance Evaluation Review}, vol.~45, no.~2, pp. 15--17, 2017.

\bibitem{BV2019}
A.~Braverman, J.~G. Dai, X.~Liu, and L.~Ying, ``Empty-car routing in
  ridesharing systems,'' \emph{Operations Research}, vol.~67, no.~5, pp.
  1437--1452, 2019.

\bibitem{CC2020}
S.~Wollenstein-Betech, I.~C. Paschalidis, and C.~G. Cassandras, ``Joint pricing
  and rebalancing of autonomous mobility-on-demand systems,'' in \emph{2020
  59th IEEE Conference on Decision and Control (CDC)}.\hskip 1em plus 0.5em
  minus 0.4em\relax IEEE, 2020, pp. 2573--2578.

\bibitem{CC2021}
------, ``Optimal operations management of mobility-on-demand systems,''
  \emph{Frontiers in Sustainable Cities}, vol.~3, p. 681096, 2021.

\bibitem{NZ2017}
A.~Nikzad, ``Thickness and competition in ride-sharing markets,''
  \emph{Available at SSRN 3065672}, 2017.

\bibitem{AZ2021}
B.~Turan and M.~Alizadeh, ``Competition in electric autonomous mobility on
  demand systems,'' \emph{IEEE Transactions on Control of Network Systems},
  2021.

\end{thebibliography}

\appendix

\section{Appendix}
\subsection{Analysis of Observation 2}
Consider player $A$'s optimization problem in \ref{eq:opt_problem}. Since $x_A^{ij}$ are functions of $p_A^{ij}$, we will substitute $x_A^{ij}$'s throughout the problem. Additionally, define $C_{ij} = \frac{1}{2}D^{ij}(1+p_B^{ij})$. Since we are solving player $A$'s optimization problem given the prices of player $B$, we treat $C_{ij}$ as constants. Now, we will convert this constrained optimization problem into an unconstrained one by constructing the Lagrangian ($\mathcal{L}$). Let us define the Lagrangian multipliers associated with some of the constraints we need : 
\begin{itemize}
    \item $K_i$ for (2b)
    \item $L_{ij}$ and $H_{ij}$ for the lower and upper bounds in (2e)
    \item $Q_{ij}$ for (2f)
\end{itemize}
Therefore, we have the following first order KKT condition :  
\begin{equation}\label{eq:p}
    \begin{split}
        &\frac{\partial \mathcal{L}}{\partial p_A^{11}} = C_{11}\left(-1 + 2p_A^{11} - p_c^{11} + p_e^{1} - K_1\right) - L_{11} + H_{11} = 0 \\
        \implies & -1 + 2p_A^{11} - p_c^{11} + p_e^{1} - K_1 = 0
    \end{split}
\end{equation}
Note that we have already argued about why $L_{11}$ and $H_{11}$ can be zero (from complementary slackness). Also, $C_{ij} \neq 0$. We can obtain a similar condition involving $r_A^{11}$ : 
\begin{equation}\label{eq:r}
    \begin{split}
        &\frac{\partial \mathcal{L}}{\partial r_A^{11}} = p_c^{11} - p_e^{1} + K_1 - Q_{11} = 0
    \end{split}
\end{equation}
Combining Eqs. \ref{eq:p} and \ref{eq:r}, we get  $p_A^{11} = \frac{1}{2}(1 + Q_{11})$. We can derive similarly for all other $p_A^{ij}$'s.

\subsection{Proof of Observation 3}
\begin{proof}
We will prove this by contradiction. Suppose, there exists a potential function $\Phi$ for this game. Let $U_A(\cdot)$ and $U_B(\cdot)$ represent the objectives of players $A$ and $B$ respectively. Then the following must hold : 
\begin{equation*}
    \begin{split}
        &\frac{\partial \Phi}{\partial p_A^{ij}} = \frac{\partial U_A}{\partial p_A^{ij}} = \frac{1}{2}D^{ij}(1+p_B^{ij})\left(1-2p_A^{ij}+p_c^{ij}-p_e^{i}\right) \\
        &\frac{\partial \Phi}{\partial p_B^{ij}} = \frac{\partial U_B}{\partial p_B^{ij}} = \frac{1}{2}D^{ij}(1+p_A^{ij})\left(1-2p_B^{ij}+p_c^{ij}-p_e^{i}\right) \\
        &\frac{\partial \Phi}{\partial r_A^{ij}} = \frac{\partial U_A}{\partial r_A^{ij}} = p_e^{i} - p_c^{ij} = \frac{\partial U_B}{\partial r_B^{ij}} = \frac{\partial \Phi}{\partial r_B^{ij}}\\
        &\frac{\partial \Phi}{\partial m_A^{i}} = \frac{\partial U_A}{\partial m_A^{i}} = - p_e^{i} = \frac{\partial U_B}{\partial m_B^{i}} = \frac{\partial \Phi}{\partial m_B^{i}}
    \end{split}
\end{equation*}
It is easy to see that $r_A^{ij}$, $r_B^{ij}$, $m_A^{i}$ and $m_B^{i}$ will appear in $\Phi(\cdot)$ as linear terms. Now, we will try to reconstruct the rest of $\Phi(\cdot)$ from the conditions above. 
\begin{equation*}
    \begin{split}
        &\frac{\partial \Phi}{\partial p_A^{ij}} =  \frac{1}{2}D^{ij}(1+p_B^{ij})\left(1-2p_A^{ij}+p_c^{ij}-p_e^{i}\right) \\
        &\implies \Phi = \sum_{i,j} \frac{1}{2}D^{ij}(1+p_B^{ij}) \int_{p_A^{ij}}(1-2p_A^{ij}+p_c^{ij}-p_e^i)dp_A^{ij} + \sum_{ij}\Psi(p_B^{ij}) \\& \quad - \sum_{i} p_e^{i} \left(m_A^i + m_B^i\right) + \sum_{i,j} (p_e^i - p_c^{ij})\left(r_A^{ij} + r_B^{ij}\right)
    \end{split}
\end{equation*}
Note that $\Psi(\cdot)$ is purely a function of $p_B^{ij}$, otherwise it would have contributed terms to $\frac{\partial \Phi}{\partial p_A^{ij}}$. This also gives us : 
\begin{equation*}
    \begin{split}
        &\frac{\partial \Phi}{\partial p_B^{ij}} = \frac{1}{2}D^{ij}\left(p_A^{ij}(1 + p_c^{ij}-p_e^i) - (p_A^{ij})^2\right) + \Psi^{'}(p_B^{ij})
    \end{split}
\end{equation*}
But, we know that $\frac{\partial \Phi}{\partial p_B^{ij}} = \frac{1}{2}D^{ij}(1+p_A^{ij})\left(1-2p_B^{ij}+p_c^{ij}-p_e^{i}\right)$. Therefore, we must have : 
\begin{equation*}
    \begin{split}
    \frac{1}{2}D^{ij}(1+p_A^{ij})&\left(1-2p_B^{ij}+p_c^{ij}-p_e^{i}\right) \\&= \frac{1}{2}D^{ij}\left(p_A^{ij}(1 + p_c^{ij}-p_e^i) - (p_A^{ij})^2\right) + \Psi^{'}(p_B^{ij})
    \end{split}
\end{equation*}
This gives us an expression for $\Psi^{'}(p_B^{ij})$. But it turns out to be a function of both $p_A^{ij}$ and $p_B^{ij}$ which is clearly a contradiction. This concludes the proof. 
\end{proof}

\subsection{Proof of Theorem 3.1}
\begin{proof}
In this proof, we will recycle the same techniques we used in the last proof. Since this is an if and only if statement, we need to prove both directions. \\
($\impliedby$) We know that the demand function $f(p_A, p_B)$ is of the form $g(p_A) + h(p_B)$. Suppose, there exists an exact potential function admitted by this demand function. Let's call it $\Phi$. Then, the following must hold :
\begin{equation*}
    \begin{split}
        &\frac{\partial \Phi}{\partial p_A^{ij}} = D^{ij}\left( g(p_A^{ij}) + h(p_B^{ij}) + g^{'}(p_A^{ij})(p_A^{ij}-p_c^{ij} + p_e^{i})\right)\\
        &\frac{\partial \Phi}{\partial p_B^{ij}} = D^{ij}\left( g(p_B^{ij}) + h(p_A^{ij}) + g^{'}(p_B^{ij})(p_B^{ij}-p_c^{ij} + p_e^{i})\right)
    \end{split}
\end{equation*}
The other partial derivatives $\frac{\partial \Phi}{\partial r_A^{ij}}$, $\frac{\partial \Phi}{\partial r_B^{ij}}$, $\frac{\partial \Phi}{\partial m_A^{i}}$ and $\frac{\partial \Phi}{\partial m_B^{i}}$ remain the same as in the earlier proof and we will reuse them directly. By reconstruction from the partial wrt $p_A^{ij}$, $\Phi$ should be of the following form : 
\begin{equation*}
    \begin{split}
        \Phi &= \sum_{i,j} D^{ij}p_A^{ij}\left(g(p_A^{ij}) + h(p_B^{ij}) + p_e^i - p_c^{ij}\right) + \sum_{i,j} \Psi(p_B^{ij})\\ &\quad- \sum_{i} p_e^{i} \left(m_A^i + m_B^i\right) + \sum_{i,j} (p_e^i - p_c^{ij})\left(r_A^{ij} + r_B^{ij}\right)
    \end{split}
\end{equation*}
This implies, $\frac{\partial \Phi}{\partial p_B^{ij}} = \Psi^{'}(p_B^{ij}) + D^{ij}p_A^{ij} h^{'}(p_B^{ij})$ which should be equal to $D^{ij}\left( g(p_B^{ij}) + h(p_A^{ij}) + g^{'}(p_B^{ij})(p_B^{ij}-p_c^{ij} + p_e^{i})\right)$. Now, since $\Psi^{'}(p_B^{ij})$ is only a function of $p_B^{ij}$, the following must be true : 
\begin{equation*}
    \begin{split}
    &p_A^{ij} h^{'}(p_B^{ij}) = h(p_A^{ij}) \\
    \implies & \frac{h(p_A^{ij})}{p_A^{ij}} = h^{'}(p_B^{ij})
    \end{split}
\end{equation*}
Therefore, $h(p)$ must be the $zero$ function or $h(p) = Cp$. Since $h(p)$ represents the effect of the competitor's price on a player's demand, $h(p) = 0$ is unrealistic. Therefore, $h(p) = Cp$. Also, note from the properties of the demand function $f(p, q)$ that $\frac{\partial f}{\partial q} \geq 0$. This implies $C > 0$ (we have already ruled out the $zero$ case). \\ \\
($\implies$) For this direction, we have to show that if there is a demand function $f(p_A, p_B) = g(p_A) + h(p_B)$ with $h(p_B) = Cp_B$ for some $C > 0$, then it admits a potential function. So, if we can come up with a valid potential function, we are done. We claim that the following is a potential function for the game for the above demand function : 
\begin{equation*}
    \begin{split}
        \Phi &= \sum_{i,j} D^{ij} \left[p_A^{ij}g(p_A^{ij}) + p_B^{ij}g(p_B^{ij}) + (p_e^{i} - p_c^{ij})\left(g(p_A^{ij}) + g(p_B^{ij})\right) \right] \\ & + C\sum_{i,j}D^{ij} p_A^{ij}p_B^{ij}  - \sum_{i} p_e^{i} \left(m_A^i + m_B^i\right) + \sum_{i,j} (p_e^i - p_c^{ij})\left(r_A^{ij} + r_B^{ij}\right)
    \end{split}
\end{equation*}
It can be easily verified that the above is a potential function for this game. This concludes the proof. 
\end{proof}

\end{document}